\documentclass[letterpaper, 10 pt, conference]{ieeeconf}
\IEEEoverridecommandlockouts                              % This command is only
\usepackage{amssymb,amsfonts,amsmath,amscd,dsfont,mathrsfs,bm}
\usepackage{graphicx,float,psfrag,epsfig,color,enumerate,hyperref}

\newtheorem{theorem}{Theorem}[section]
\newtheorem{assumption}{Assumption}

\newtheorem{lemma}[theorem]{Lemma}

\newtheorem{corollary}[theorem]{Corollary}

\newtheorem{definition}[theorem]{Definition}
\newtheorem{remark}[theorem]{Remark}

\def\prob{{\mathbb P}}
\def\root{{\o}}
\def\nat{{\mathds N}}
\def\hxi{\widehat{\xi}}
\def\errI{\eta^{\textup{\tiny I}}}
\def\errII{\eta^{\textup{\tiny II}}}

\def\hsigma{\widehat{\sigma}}

\newcommand{\bomega}[1]{\underline{\omega}^{#1}}
\newcommand{\balf}{\overline{\alpha}}
\newcommand{\bbet}{\overline{\beta}}
\newcommand{\bQ}{\overline{Q}}
\def\ind{{\mathbb I}}
\def\cX{\mathcal{X}}
\def\cS{\mathcal{S}}
\def\cM{\mathcal{M}}
\def\cG{\mathcal{G}}
\def\cE{\mathcal{E}}
\def\cV{\mathcal{V}}
\def\cD{\mathcal{D}}
\def\naturals{\mathbb{N}}
\def\taup{\tau_{\rm p}}
\def\taud{\tau_{\rm d}}
\def\pe{\prob_{\rm err}}

\def\brho{{\overline{\rho}}}
\def\eps{\epsilon}

\title{Subexponential convergence for information aggregation on regular trees}

\author{Yashodhan Kanoria and Andrea Montanari
\thanks{Department of Electrical Engineering and Department of Statistics, Stanford University.
 Supported by 3Com Corporation Stanford Graduate Fellowship, and NSF
 grants CCF-0743978 and CCF-0915145.}}

\begin{document}

\maketitle
\begin{abstract}
We consider the decentralized binary hypothesis testing problem
on trees of bounded degree and increasing depth.
For a regular tree of depth $t$ and branching factor $k\ge 2$,
we assume that the leaves have access to independent and identically
distributed noisy observations of the `state of the world' $s$.
Starting with the leaves, each node makes a decision in a finite alphabet
$\cM$, that it sends to its parent in the tree. Finally, the root decides between the two possible
states of the world based on the information it receives.

We prove that the error probability vanishes only subexponentially in
the number of available observations, under quite general hypotheses.
More precisely the case of binary messages, decay is subexponential
for any decision rule. For general (finite) message alphabet $\cM$,
decay is subexponential for `node-oblivious' decision rules,
that satisfy a mild irreducibility condition.
In the latter case, we propose a family of decision rules
with close-to-optimal asymptotic behavior.
\end{abstract}

\section{Introduction}

Let $G = (V,E)$ be a (possibly infinite) network rooted at node
$\root$. Assume that independent and identically distributed noisy
observations of an hidden random variable  $s\in\{0,1\}$  are
available at a subset $U\subseteq V$ of the vertices. Explicitly, each
$i\in U$ has access to a private signal $x_i\in\cX$ where
$\{x_i\}_{i\in U}$ are independent and identically distributed,
conditional on $s$. The `state of the world' $s$ is drawn from
a prior probability distribution $\pi = (\pi_0,\pi_1)$.
The objective is to aggregate information about $s$ at the root node
under communication constraints encoded by the network structure,
while minimizing the error probability at $\root$.

We ask the following question:

\begin{quote}
\emph{How much does the error probability at the
root node $\root$ increase due to these communication constraints?}
\end{quote}

In order to address this question,
 consider a sequence of information aggregation problems indexed by
$t$.
%We model communication constraints as follows:
Information is revealed
in a subset of the vertices $U_t \subseteq V$. There are $t$ rounds in
which information aggregation occurs. In each round, a subset of the nodes in $V$ make `decisions' that are broadcasted to their neighbors.
In the initial round, nodes $i \in U_t$
with distance $d(\root,i)=t$ (with $d(\,\cdot\, ,\,\cdot\,)$ being the graph distance)
broadcast a decision $\sigma_i \in \cM$ to their neighbors, with $\cM$
a finite alphabet. In the next round,
nodes $i \in V$ with distance $d(\root,i)=t-1$ broadcast a decision $\sigma_i \in \cM$ to their neighbors. And so on, until the neighbors of $\root$ announce their decisions in round $t$. Finally, the root makes its decision. The  decision of any node $i$ is a function of decisions of $i$'s neighbors in earlier rounds, and, if $i\in U$,
on the private signal $x_i$ received by $i$.

Clearly, the root can possibly access only the private
information available at nodes $i\in V$ with $d(\root,i)\le t$
(with $d(\,\cdot\, ,\,\cdot\,)$ the graph distance). We can therefore
assume, without loss of generality, that $U_t\subseteq
\{i\in V:\, d(\root,i)\le t\}$.  It is convenient to think of $U_t$ as
the \emph{information horizon} at time $t$.

Consider first the case in which communication is unconstrained.
This can be modeled by considering the graph with vertices
$V=\{\root,1,2,3,\dots\}$ and edges $E= \{(\root,1),(\root,2),(\root,3),\dots\}$.
In other words, this is a star network, with the root at the
center. Without loss of generality, we
take $U_{t} = \{1,\dots,|U_t|\}$, with $|U_t|\uparrow \infty$
as $t\to\infty$.

A simple procedure for information aggregation
would work as follows. Each node $i$ computes the log-likelihood
ratio (LLR) $\ell(x_i)$ corresponding to the  observed signal, and quantizes
it to a value $\sigma_i$. The root adds up the quantized  LLRs and
decides on the basis of this sum.
It follows from basic large deviation theory \cite{DemboZeitouni}
that, under mild regularity assumptions, the error probability decreases
exponentially in the number of observations
\begin{eqnarray}
\prob\{\sigma_{\root} \neq s\} = \exp\big\{-\Theta(|U_t|)\}\, .
\end{eqnarray}
This result is extremely robust:

\noindent{\bf $(1)$}~It holds for any non-trivial
alphabet $|\cM|\ge 2$;

\noindent{\bf $(2)$}~Using concentration-of-measure
arguments \cite{Ledoux,Dubhashi} it is easy to generalize it
to families of weakly dependent observations
\cite{KanoriaMontanariWeakly};

\noindent{\bf $(3)$}~It can be generalized to network structures $G$ with weak
communications constrains. For instance, \cite{TTW_bounded_height}
proved that the error probability decays exponentially in the number
of observations for trees of bounded depth. The crucial observation
here is that such networks have large degree diverging with the number
of vertices. In particular, for a tree of depth $t$, the maximum
degree is at least $n^{1/t}$.

\vspace{0.1cm}

At the other extreme,  Hellmann and Cover \cite{HellmannCover} considered the case of a
line network. In our notations, we have $V = \{\root,1,2,3,\dots\}$,
$E=\{(\root,1),(1,2),(2,3),\dots\}$, and $U_t = \{1,2,\dots,t\}$.
In \cite{HellmannCover} they proved that, as long as the LLRs are
bounded (namely $|\ell(x_i)|\le C$ almost surely for some constant
$C$), and the decision rule is independent of the node,
 the error probability  remains bounded away from $0$ as
$t\to\infty$.

If the decision rule is allowed to depend on the node, the error
probability  can vanish as $t\to\infty$ provided $|\cM|\ge 3$
\cite{Cover,Koplowitz}. Despite this, even if the probability of error
decays to $0$, it does so much more slowly than for
highly connected networks. Namely, Tay, Tsitsiklis and
Win \cite{TTW_subexp_tandem} proved that
\begin{eqnarray}
\prob\{\sigma_{\root}\neq s\} = \exp\big\{-O(|U_t|^{\rho})\big\}\label{eq:Subexp}
\end{eqnarray}
for some $\rho<1$.
 In other words, the communication constraint is so severe that,
after $t$ steps, the amount of information effectively used by
the root  is equivalent to a vanishingly small fraction of the one within
the `information horizon'.

These limit cases naturally lead to the general question: Given a rooted network
$(G,\root)$, a sequence of information horizons $\{U_t\}_{t\ge 1}$
and a finite alphabet $\cM$, can information be aggregated at the root
in such a way that the error probability decays exponentially in
$|U_t|$? The question is wide open, in particular for networks of
with average degree bounded or increasing slowly
(e.g. logarithmically) with the system size.

Networks with moderate degree arise in a number of practical
situations. Within decentralized detection applications,  moderate degree
is a natural assumption for interference-limited wireless networks.
In particular, systems in which a single root node communicates with a
significant fraction of the sensors are likely to scale poorly because
of interference at the root. Standard models for wireless ad hoc
networks \cite{GuptaKumar} are indeed based on random geometric graphs
whereby each node is connected to a logarithmic number of neighbors.

A different domain of applications for models of decentralized
decision making is social learning \cite{ADLOzdaglar}. In this case, each node
corresponds to  an agent, and the
underlying graph is the social network across which
information is exchanged. Also in this case, it is reasonable to
assume that each agent has a number of neighbors which is bounded,
or diverges slowly as the total number of agents grows.
In many graph-theoretic models of social networks \cite{Newman},
although a small number of  nodes can have large degree,
the average degree is bounded or grows
logarithmically with the network size.

Given the slow progress with extreme network structures (line networks
and highly-connected networks), the study of general moderate degree
networks appears extremely challenging. In this paper we
focus on regular trees. More precisely, we let $G$ be the (infinite)
regular tree with branching factor $k$, rooted at $\root$
(each node has $k$ descendants and, with the exception of the root,
one parent). The information horizon  $U_t$ is formed by all the nodes
at distance $t$ from the root, hence $|U_t|=k^t$.
Under a broad set of assumptions, we prove that the probability of
error decays subexponentially in the size of the information set,
cf. Eq.~(\ref{eq:Subexp}), where $\rho=\rho_\cM<1$ depends on the
size of the alphabet $|\cM|=m$.

More precisely, we establish subexponential convergence in the following cases:
\begin{enumerate}
\item For binary messages $|\cM|=2$ and any choice of the decision
  rule. In fact, we obtain a precise characterization of the smallest possible
  error probability in this case.
\item For general message alphabet $3\le |\cM|<\infty$ provided
the decision rule does not depend on the node, and satisfies a mild
`irreducibility' condition (see Section \ref{subsec:subexp_decay_general} for a definition).
\end{enumerate}

In the latter case, one expects that exponential convergence is
recovered as the message set gets large. Indeed we prove that
the optimal exponent in Eq.~(\ref{eq:Subexp}) obeys
\begin{eqnarray}
 1-\frac{C_1}{|\cM|}\le\rho_{\cM} \le 1-\exp\big\{-C_2|\cM|\big\} \,
 .\label{eq:LargeAlphabet}
\end{eqnarray}
The upper bound follows from our general proof for irreducible
decision rules, while the lower bound is obtained by
 constructing an explicit decision rule that achieves it.

Our investigation leaves several interesting open problems.
First, it would be interesting to
compute the optimal exponent $\rho = \rho(k,\cM)$ for given degree of
the tree and size of the alphabet. Even the behavior of the exponent
for large alphabet sizes is unknown at the moment (cf. Eq.~\eqref{eq:LargeAlphabet}).
Second, the question of characterizing the performance limits of general, node-dependent decision rules remains open for $|\cM|\geq 3$.
Third, it would be interesting to understand the case where non-leaf nodes also get
private signals, e.g., $U_t = \{i: i \in V, d(\root, i)\leq t\}$.
Finally, this paper focuses on tree of bounded degree. It
would be important to explore generalization to other graph
structures, namely trees with slowly diverging degrees (which could be
natural models for the local structure of preferential attachment
graphs \cite{SaberiPA}),  and loopy graphs. Our current results
can be extended to trees of diverging degree only in the case of
binary signals. In this case we obtain that  the probability of error
is subexponential
\begin{eqnarray}
\prob\{\sigma_{\root}\neq s\} = \exp\big\{-o(|U_t|)\}
\end{eqnarray}
as soon as the degree is sub-polynomial, i.e.  $k= o(n^{a})$ for all $a>0$.

\vspace{0.1cm}

The rest of the paper is organized as follows:
Section \ref{sec:Model} defines formally the model for information
aggregation. Section \ref{sec:Binary} presents our results for binary
messages $|\cM|=2$. Section \ref{sec:NodeOblivious}
treats the case of decision rules that do not depend on the node, with general $\cM$.

\section{Model Definition}
\label{sec:Model}

As mentioned in the introduction, we assume the network
$G=(V,E)$ to be an (infinite) rooted  $k$-ary tree, i.e. a tree
whereby each node has $k$ descendants and one parent (with the
exception of the root, that has no parent).
Independent noisy observations (`private signals')
of the state of the world $s$
are provided to the nodes at all the nodes at $t$-th generation
$U_t= \{i\in V\,:\, d(\root,i) = t\}$. These will be also referred to
as the `leaves'. Define $n \equiv |U_t| = k^t$.
Formally, the state of the world $s\in \{0,1\}$ is drawn according to
the prior $\pi$ and for each $i\in U_t$ an independent
observation $x_i \in \cX$ is drawn with probability distribution
$p_0(\,\cdot\,)$ (if $s=0$) or $p_1(\,\cdot\,)$ (if $s=1$).
For notational simplicity we assume that $\cX$ is finite, and that
$p_0(x)$, $p_1(x)>0$ for all $x\in\cX$. Also, we exclude degenerate
cases by taking $\pi_0,\pi_1>0$.
We refer to the refer to the two events $\{s=0\}$ and $\{s=1\}$ as the
hypotheses $H_0$ and $H_1$. 

In round 0, each leaf $i$ sends a message $\sigma_i \in \cM$ to its parent at level 1. In round 1, the each node $j$ at level 1 sends a message $\sigma_j \in \cM$ to its parent at level 2.
Similarly up to round $t$. Finally, the root node $\root$ makes a decision $\sigma_\root \in \{0,1\}$ based on the $k$ messages it receives. The objective is to minimize $\pe \equiv \prob(\sigma_\root \neq s)$. We call a set of decision rules \emph{optimal} if it minimizes $\pe$.

We will denote by $\partial i$ the set of children of node $i$. We denote the probability of events under $H_0$ by
$\prob_0(\cdot)$, and the probability of events under $H_1$ by
$\prob_1(\cdot)$.
 Finally, we denote by $f_i$ the decision rule at node $i$ in the tree.  If $i$ is not a leaf node and $i\neq \root$, then $f_i: \cM^k \rightarrow \cM$. The root makes a binary decision $f_\root: \cM^k \rightarrow \{0,1\}$. If $i$ is a leaf node, it maps its private signal to a message, $f_i: \cX \rightarrow \cM$.  In general, $f_i$'s can be randomized.

\section{Binary messages}
\label{sec:Binary}

In this section, we consider the case $\cM=\{0,1\}$, i.e., the case of binary messages.

Consider the case $\pi_0=\pi_1=1/2$, $\cX=\{0,1\}$ and $p_s(x) = (1-\delta) \ind(x=s) + \delta \ind(x\neq s)$ for $s=0, 1$; where  $\delta \in (0,1/2)$. Define the majority decision rule  at non-leaf node $i$ as follows: $\sigma_i$ takes the value of the majority of $\sigma_{\partial i}$ (ties are broken uniformly at random). %Thus, if $k=9$ and exactly 6 of $i$'s children sent a message 1 (hence 3 sent a message 0), then under the majority rule $\sigma_i = 1$.

It is not hard to see that if we implement majority updates at all non-leaf nodes, we achieve
\begin{align}
\prob_{\rm maj}(\sigma_\root \neq s) = \exp\left \{ -\Omega\left( \lfloor (k+1)/2 \rfloor^t  \right) \right\}
\end{align}
Note that this is an upper bound on error probability under majority updates.

Our main result shows that, in fact, this is essentially the best that can be achieved.
\begin{theorem}
\label{thm:binary_lower_bound}
Fix the private signal distribution, i.e., fix $p_0(\cdot)$ and $p_1(\cdot)$. There exists $C< \infty$ such that for all $k \in \nat$ and $t \in \nat$, for any combination of decision rules at the nodes, we have
\begin{align}
\prob(\sigma_\root \neq s) \geq  \exp\left \{ -C \left ( \frac{k+1}{2}   \right)^t \right\}
\end{align}
\end{theorem}

In particular, the error probability decays subexponentially in the number of private signals
$n = k^t$, even with the optimal protocol. 

\subsection{Proof of Theorem \ref{thm:binary_lower_bound}}
 We prove the theorem for the case $\pi_0=\pi_1=1/2$, $\cX=\{0,1\}$ and $p_s(x) = (1-\delta) \ind(x=s) + \delta \ind(x\neq s)$ for $s=0, 1$; where $\delta \in (0,1/2)$. The proof easily generalizes to arbitrary $\pi, \cX, p_0$ and $p_1$.

Also, without loss of generality we can assume that, for every node $i$,
\begin{align}
\label{eq:LRorder}
\frac{\prob(s=1|\sigma_i=1)}{\prob(s=0|\sigma_i=1)} \geq
\frac{\prob(s=1|\sigma_i=0)}{\prob(s=0|\sigma_i=0)}
\end{align}
(otherwise simply exchange the symbols and modify the decision rules accordingly).

Denote by $\errI_i$ the (negative) logarithm of the `type I error' in $\sigma_i$, i.e. $\errI_i\equiv -\log(\prob(s=0, \sigma_i=1))$.
Denote by $\errII_i$ the (negative) logarithm of the `type II error' in $\sigma_i$, i.e. $\errII_i\equiv - \log(\prob(s=1, \sigma_i=0))$.

The following is the key lemma in our proof of Theorem \ref{thm:binary_lower_bound}.

\begin{lemma}
\label{lemma:binary_key_lemma}
Given $\delta>0$, there exists $C\equiv C(\delta)>0$ such that for any $k$ we have the following: There exists an optimal set of decision rules such that for any node $i$ at level $\tau \in \nat$,
\begin{align}
\label{eq:product_errIandII}
\errI_i \errII_i \leq C^2 ((k+1)/2)^{2\tau} \, .
\end{align}
\end{lemma}

\begin{proof}[Proof of Theorem \ref{thm:binary_lower_bound}]
Applying Lemma \ref{lemma:binary_key_lemma} to the root $\root$, we
see that $\min(\errI_\root, \errII_\root) \leq C((k+1)/2)^t$.
The result follows immediately.
\end{proof}

Lemma \ref{lemma:binary_key_lemma} is proved using the fact that there is an optimal
set of decision rules that correspond to deterministic likelihood ratio tests (LRTs)
at the non-leaf nodes.

\begin{definition}
Choose a node $i$. Fix the decision functions of all descendants of $i$. Define
$L_i(\sigma_{\partial i}) = \prob(H_1|\sigma_{\partial i}) / \prob(H_0|\sigma_{\partial i})$. \\ 
a) The decision function  $f_i$ is a \emph{monotone deterministic likelihood ratio test} if:\\
(i) It is deterministic.\\
(ii) There is a threshold $\theta$ such that
\begin{align*}
\prob(f_i=1, L_i<\theta)=0\\
\prob(f_i=0, L_i>\theta)=0
\end{align*}

b) The decision function  $f_i$ is a \emph{deterministic likelihood ratio test} if either $f_i$ or $f_i^{\rm c}$ is a monotone deterministic likelihood ratio test. Here $f_i^{\rm c}$ is the Boolean complement of $f_i$.
\end{definition}

The next lemma is an easy consequence of a beautiful result of Tsitsiklis \cite{T_LRQ_extremal}. Though we state it here only for binary message alphabet, it easily generalizes to arbitrary finite $\cM$.
%{\bf [Y: define LRTs and prove Lemma below.]}
\begin{lemma}
\label{lemma:LRToptimal}
There is a set of monotone deterministic likelihood ratio tests at the nodes that achieve the
minimum possible $\prob(\sigma_\root \neq s)$.
\end{lemma}
\begin{proof}
Consider a set of decision rules that minimize $\prob(\sigma_\root \neq s)$.

Fix the rule at every node except node $i$ to the optimal one. Now, the distributions $\prob_0(\sigma_{\partial i})$ and $\prob_1(\sigma_{\partial i})$ are fixed. Moreover,
$\prob(\sigma_\root \neq s)$ is a linear function of $q(f_i)\equiv(\prob_0(\sigma_i), \prob_1(\sigma_i))$, where $\prob_s(\sigma_i)$ denotes the distribution of $\sigma_i$ under hypothesis $H_s$. The set $\bQ$ of achievable $q$'s is clearly convex, since randomized $f_i$ is allowed. From \cite[Proposition 3.1]{T_LRQ_extremal}, we also know that $\bQ$ is compact.
Thus, there exists an extreme point of $\bQ$ that minimizes $\prob(\sigma_\root \neq s)$.
Now \cite[Proposition 3.2]{T_LRQ_extremal} tells us that any extreme point of $\bQ$ can be achieved by a deterministic LRT. Thus, we can change $f_i$ to a deterministic LRT without increasing $\prob(\sigma_\root \neq s)$. If $f_i$ is not monotone (we know that $i\neq \root$ in this case), then we do $f_i\leftarrow f_i^{\rm c}$ and $f_j(\sigma_i, \sigma_{\partial j \backslash i}) \leftarrow f_j(\sigma_i^{\rm c}, \sigma_{\partial j \backslash i})$.  Clearly, $\prob(\sigma_\root \neq s)$ is unaffected by this transformation, and $f_i$ is now a monotone rule.

We do this at each of the nodes sequentially, starting at level $0$, then covering level $1$ and so on until the root $\root$. 
Thus, we change (if required) each decision rule to a monotone deterministic LRT without increasing $\prob(\sigma_\root \neq s)$. The result follows.
\end{proof}

Clearly, if $f_i$ is a monotone LRT, Eq.~\eqref{eq:LRorder} holds. In fact, we argue that there is a set of deterministic monotone LRTs with strict inequality in Eq.~\eqref{eq:LRorder}, i.e., such that
\begin{align}
\label{eq:LRorderstrict}
\frac{\prob(s=1|\sigma_i=1)}{\prob(s=0|\sigma_i=1)} >
\frac{\prob(s=1|\sigma_i=0)}{\prob(s=0|\sigma_i=0)}
\end{align}
holds for all $i$, that are optimal.

Eq.~\eqref{eq:LRorder} can only be written when $\prob(\sigma_i=0)>0$ and $\prob(\sigma_i=1)>0$. Consider a leaf node $i$. Without loss of generality we can take $\sigma_i = x_i$ for each leaf node $i$ (since any other rule can be `simulated' by the concerned level 1 node). So we have $\prob(\sigma_i=0)>0$ and $\prob(\sigma_i=1)>0$,  Eq.~\eqref{eq:LRorderstrict} holds and $f_i$ is a deterministic LRT. We can ensure these properties inductively at all levels of the tree by moving from the leaves towards the root.
Consider any node $i$. If $\prob(\sigma_i=0)=0$, then $i\neq \root$ (else $\pe=1/2$) and the parent of $i$ is ignoring the constant message received from $i$. We can do at least as well by using any non-trivial monotone deterministic LRT at $i$. Similarly, we can eliminate 
$\prob(\sigma_i=1)=0$. If $\prob(\sigma_i=0)>0$ and $\prob(\sigma_i=1)>0$, then Eq.~\eqref{eq:LRorderstrict} must hold for any monotone deterministic LRT  $f_i$, using the inductive hypothesis.
%{\bf [Y: I'm not sure we really need the fact that deterministic LRTs suffice in our proof of Lemma \ref{lemma:binary_key_lemma}.]}

\begin{definition}
Let $\balf$ and $\bbet$ be binary vectors of the same length $\tau$. We say
$\balf \succeq \bbet$ if $\alpha_i \geq \beta_i$ for all $i\in \{1,2, \ldots, \tau\}$.
\end{definition}
We now prove Lemma \ref{lemma:binary_key_lemma}.

\begin{proof}[Proof of Lemma \ref{lemma:binary_key_lemma}]

From Lemma \ref{lemma:LRToptimal} and Eq.~\eqref{eq:LRorderstrict}, we can restrict attention to monotone deterministic LRTs satisfying Eq.~\eqref{eq:LRorderstrict}.

We proceed via induction on level $\tau$. For any leaf node $i$,
we know that $\errI_i = \errII_i= -\log(\delta/2)$. Choosing $C= -\log(\delta/2)$, Eq.~\eqref{eq:product_errIandII} clearly
holds for all nodes at level $0$. Suppose Eq.~\eqref{eq:product_errIandII} holds
for all nodes at level $\tau$. Let $i$ be a node at level $\tau+1$. Let its children be
$\partial i = \{c_1, c_2, \ldots, c_k\}$. Without loss of generality, assume
\begin{align}
\errI_{c_1} \geq \errI_{c_2} \geq \ldots \geq \errI_{c_k}
\label{eq:errIorder}
\end{align}

\noindent{\bf Claim:} We can also assume
\begin{align}
\errII_{c_1} \leq \errII_{c_2} \leq \ldots \leq \errII_{c_k}
\label{eq:errIIorder}
\end{align}
\noindent Proof of Claim:
Suppose, instead, $\errII_{c_1} > \errII_{c_2}$ (so $c_1$ is doing better than $c_2$ on both types of error).  We can use the protocol
on the subtree of $c_1$ also on the subtree of $c_2$. Call the message of $c_2$ under this modified protocol $\hsigma_{c_2}$.
 %and the corresponding error probabilities $\herrI_{c_2}$ and $\herrI_{c_2}$.
 Since, $\errI_{c_1}\geq \errI_{c_2}$ and $\errII_{c_1} \geq \errII_{c_2}$ (both types of error have only become less frequent), 
there exists a randomized function $F: \{0,1\} \rightarrow \{0,1\}$, such that $\prob_s(F(\hsigma_{c_2})=1)= \prob_s(\sigma_{c_2}=1)$ for $s=1,2$.
Thus, node $i$ can use $f_i(\sigma_{c_1}, F(\hsigma_{c_2}), \sigma_{c_3}, \ldots, \sigma_{c_k})$ to achieve the original values of $\errI_{c_2}$ and $\errII_{c_2}$, where $f_i$ is decision rule being used at $i$ before. Clearly, the error probabilities at $i$, and hence at the root, stay unchanged with this. Thus, we can safely assume $\errII_{c_1} \leq \errII_{c_2}$. Similarly, we can assume $\errII_{c_i} \leq \errII_{c_{i+1}}$ for $i=2,3, \ldots,k-1$. Clearly, our transformations retained the property that nodes at levels $\tau+1$ and below use deterministic LRTs satisfying Eq.~\eqref{eq:LRorderstrict}. Similar to our argument for Eq.~\eqref{eq:LRorderstrict} above, we can make
appropriate changes in the decision rules at levels above $\tau+1$ so that they also use deterministic LRTs satisfying Eq.~\eqref{eq:LRorderstrict}, without increasing error probability. 
 This proves the claim.

Recall that $f_i: \{0,1\}^k \rightarrow \{0,1\}$ is the decision rule at node $i$.  Assume the first bit in the input corresponds to $\sigma_{c_1}$, the second corresponds to $\sigma_{c_2}$, and so on. Using Lemma \ref{lemma:LRToptimal}, we can assume that $f_i$ implements a deterministic likelihood ratio test. Define the $k$-bit binary vectors $\bomega{0} = (111 \ldots 1)$, $\bomega{1} = (011 \ldots 1)$, \ldots, $\bomega{k} = (00 \ldots 0)$. From  Lemma \ref{lemma:LRToptimal} and Eq.~\eqref{eq:LRorderstrict}, it follows that $f_i(\bomega{j}) = \ind(j< j_0)$ for some $j_0 \in \{0, 1, \ldots, k, k+1 \}$.

\noindent{\bf Claim:} Without loss of generality, we can assume that $j_0 \neq 0$ and $j_0 \neq k+1$.

\noindent{Proof of Claim:}
Suppose $j_0=0$. It follows from Lemma \ref{lemma:LRToptimal} and Eq.~\eqref{eq:LRorderstrict} that $f_i(\sigma_{\partial i}) =0$ for every possible $\sigma_{\partial i}$. If $i=\root$ then we have $\pe\geq 1/2$. Suppose $i \neq \root$.  Then $\sigma_i$ is a constant and is ignored by the parent of $i$. We cannot do worse by
using an arbitrary non-trivial decision rule at $i$ instead. (The parent can always continue
to ignore $\hsigma_i$.)
The case $j_0 = k+1$ can be similarly eliminated. This proves the claim.

Thus, we can assume $j_0 \in \{ 1, \ldots, k \}$ without loss of generality. Now $\bomega{} \succeq \bomega{j_0-1}$ contribute to type I error and $\bomega{} \preceq \bomega{j_0}$ contribute to type II error. It follows that
\begin{align}
\errI_i &\leq \sum_{j=j_0}^k \errI_{c_j} \leq (k-j_0+1) \errI_{c_{j_0}} \label{eq:errI_ub} \, ,\\
\errII_i &\leq \sum_{j=1}^{j_0} \errII_{c_j} \leq j_0 \errII_{c_{j_0}} \label{eq:errII_ub} \, ,
\end{align}
where we have used the ordering on the error exponents (Eqs.~\eqref{eq:errIorder} and \eqref{eq:errIIorder}).
Eqs.~\eqref{eq:errI_ub} and \eqref{eq:errII_ub} lead immediately to
\begin{align}
\errI_i / \errI_{c_{j_0}} + \errII_i / \errII_{c_{j_0}} \leq (k+1) \,.
\label{eq:sum_err_ratios_ub}
\end{align}
Now, for any $x,y \geq 0$, we have $x+y \geq 2 \sqrt{xy}$. Plugging $x=\errI_i / \errI_{c_{j_0}}$
 and $y=\errII_i / \errII_{c_{j_0}}$, we obtain from Eq.~\eqref{eq:sum_err_ratios_ub}
\begin{align}
\errI_i \errII_i  \leq \left( \frac{k+1}{2}\right )^2 \errI_{c_{j_0}}\errII_{c_{j_0}} \,.
%\label{eq:}
\end{align}
By our induction hypothesis $\errI_{c_{j_0}}\errII_{c_{j_0}} \leq  C^2 ((k+1)/2)^{2\tau}$.
Thus, $\errI_{i}\errII_{i} \leq  C^2 ((k+1)/2)^{2(\tau+1)}$ as required. Induction completes the proof.
\end{proof}

\section{`Node-oblivious' rules with non-binary messages}
\label{sec:NodeOblivious}

In this section we allow a general finite message alphabet $\cM$ that need not be binary. However, we restrict attention to the case of \emph{node-oblivious} rules: The decision rules $f_i$ at all nodes in the tree, except the leafs and the root, must be the same. We denote this `internal node' decision rule by $f:\cM^k \rightarrow \cM$. Also, the decision rules used at each of the leaf nodes should be same. We denote the leaf decision rule by $g: \cX \rightarrow \cM$.
The decision rule at the root is denoted by $h=f_\root: \cM^k \rightarrow \{0,1\}$.
We call such $(f,g,h)$ a node-oblivious decision rule vector.

Define $m \equiv |\cM|$. In Section \ref{subsec:node-oblivious_efficient_scheme}, we present a scheme that achieves
\begin{align}
\prob(\sigma_\root \neq s) =  \exp \left \{-\Omega \Big (\,  \big \{k  \left (1-1/m\right ) \big \}^t   \, \Big) \right \}\, ,
\end{align}
when the error probability in the private signals is sufficiently small.
Next, under appropriate assumptions, we show that the decay of error probability must be sub-exponential in the number of private signals $k^t$.

\subsection{An efficient scheme}
\label{subsec:node-oblivious_efficient_scheme}

For convenience, we label the messages as
\begin{align}
\cM=\left \{\frac{-m+1}{2}\, , \frac{-m+3}{2}\, , \ldots, \frac{m-1}{2}\right \}
\end{align}
The labels have been chosen so as to be suggestive (in a quantitative sense, see below) of the inferred log-likelihood ratio. Further, we allow the messages to be treated as real numbers (corresponding to their respective labels) that can be operated on. In particular, the quantity $S_i \equiv \sum_{c \in \partial i} \sigma_c$ is well defined for a non-leaf node $i$.

The node-oblivious decision rule we employ at a non-leaf node $i \neq \root$ is
\begin{align}
f(\sigma_{\partial i}) = \left \{
\begin{array}{ll}
\left \lfloor \frac{S_i/k + (m-1)/2}{1-1/m} \right \rfloor - \frac{m-1}{2}\, ,& \textup{if } S_i \leq 0 \phantom{\bigg (}\\[4pt]
\left \lfloor \frac{S_i/k - (m-1)/2}{1-1/m}\right \rfloor + \frac{m-1}{2}\, ,& \textup{if } S_i > 0
\end{array}
\right .
\label{eq:good_decision_rule}
\end{align}
Note that the rule is symmetric with respect to a inversion of sign, except that $S_i=0$ is mapped to the message $1/2$ when $m$ is even.

The rule $g(x_i)$ used at the  leafs  is simply $g(1) = (m-1)/2$ and $g(0) = -(m-1)/2$. The decision rule at the root is
\begin{align}
h(\sigma_{\partial \root}) =  \left \{
\begin{array}{ll}
1 \, ,& \textup{if } S_\root \geq 0\\
0 \, ,& \textup{otherwise.}
\end{array}
\right .
\label{eq:root_good_dec_rule}
\end{align}
If we associate $H_0$ with negative quantities,  and $H_1$ with positive quantities, then
again, the rule at the leafs is symmetric, and the rule at the root is essentially symmetric (except for the case $S_\root=0$).

%The node-oblivious decision rule we employ at a non-leaf node $i \neq \root$ is
%\begin{align}
%f_i(\sigma_{\partial i}) = \left \{
%\begin{array}{ll}
%\frac{-m+1}{2}, & \textup{for }S_i \leq k \left (\frac{-m+1}{2} + 1\cdot \left(1-\frac{1}{m}\right )  \right )\\
%\frac{-m+3}{2}, & \textup{for }S_i \in \Big ( k \left (\frac{-m+1}{2} + 1\cdot \left(1-\frac{1}{m}\right )  \right ), k \left (\frac{-m+1}{2} + 2\cdot \left(1-\frac{1}{m}\right )  \right ) \Big ]\\
%\vdots & \\
%0, &  \textup{for }S_i \in \Big ( - k \left (\frac{m-1}{2m} \right ),
% k \left (\frac{m-1}{2m} \right )\Big )\\
%1, &  \textup{for }S_i \in \Big [  k \left (\frac{m-1}{2m} \right ),
% k \left (\frac{m-1}{2m}  + 1- \frac{1}{m}\right )\Big )\\
% \vdots & \\
%\frac{-m+3}{2}, & \textup{for }S_i \in \Big ( k \left (\frac{-m+1}{2} + 1\cdot \left(1-\frac{1}{m}\right )  \right ), k \left (\frac{-m+1}{2} + 2\cdot \left(1-\frac{1}{m}\right )  \right ) \Big ]\\
%\end{array}
%\right .
%\end{align}

\begin{lemma}
 Consider the node-oblivious  decision rule vector $(f,g,h)$ defined above. For $k \geq 2$ and $m \geq 3$, there exists
 %$C\equiv C(m,k)< \infty$ and
 $\delta_0 \equiv \delta(m,k)> 0$ such that the following is true for all $\delta < \delta_0$:

\noindent (i) Under $H_0$, for node $i$ at level $\tau \geq 0$, we have
\begin{align}
-\log \prob\big [\sigma_i=-(m-1)/2 + l\big ] \geq  (l/m) \big \{k  \left (1-1/m\right ) \big \}^\tau
\end{align}
for $l = 1, 2, \ldots, m-1$.

\noindent (ii) Under $H_1$, for node $i$ at level $\tau \geq 0$, we have
\begin{align}
-\log \prob\big [\sigma_i=(m-1)/2 - l\big ] \geq  (l/m) \big \{k  \left (1-1/m\right ) \big \}^\tau
\end{align}
for $l = 1, 2, \ldots, m-1$.
\label{lemma:good_decision_decay}
\end{lemma}

\begin{proof}
We prove (i) here. The proof of (ii) is analogous.

Assume $H_0$. Define $\gamma \equiv k  \left (1-1/m\right )$ and $C\equiv k \log m/ (k-1) $. We show that, in fact, for suitable choice of $\delta_0$ the following holds: If $\delta < \delta_0$, then for any node $i$ at any level $\tau \geq 0$,
\begin{align}
-\log \prob\big [\sigma_i=-(m-1)/2 + l\big ] \geq  \nonumber\\
(l/m) \gamma^\tau + C
\label{eq:letter_prob_decay}
\end{align}

We proceed by induction on $\tau$. Consider $i$ at level $\tau=0$. We have
$\prob_0\big [\sigma_i=-(m-1)/2 + l\big ] = 0$ for $l=1,2, \ldots, m-2$ and
$\prob_0 \big [\sigma_i=(m-1)/2] = \delta$. Choosing $\delta_0 \equiv  \exp(-1 - C)$,
we can ensure that Eq.~\eqref{eq:letter_prob_decay} holds at level $0$. Note that for %$m \gg 1$,
$k \gg 1$, we have $\delta_0 \approx 1/(em)$.

Now suppose Eq.~\eqref{eq:letter_prob_decay} holds at level $\tau$. Consider node $i$ at level $\tau+1$. From Eq.~\eqref{eq:good_decision_rule}, for $\sigma_i = -(m-1)/2 + l$ we need
\begin{align}
S_i \geq k[ -(m-1)/2 + l(1-1/m)]
\label{eq:Si_lb}
\end{align}
For every $\sigma_{\partial i}= (-(m-1)/2 + l_1, -(m-1)/2 + l_2, \ldots, -(m-1)/2 + l_k )$ such that Eq.~\eqref{eq:Si_lb} holds, we have $\sum_{j=1}^k l_j \geq k l (1-1/m) $. Thus,
\begin{align}
\prob_0 (\sigma_{\partial i}) &\leq \exp\left(- kC - (1/m)\gamma^\tau\sum_{j=1}^k l_j \right ) \nonumber\\
&\leq \exp\left(- kC - (1/m)l \gamma^{\tau+1} \right )
\end{align}
Obviously, there are at most $m^k$ such $\sigma_{\partial i}$.
Thus,
\begin{align*}
&\prob_0 [\sigma_i = -(m-1)/2 + l]  \\
\leq \;& m^k \exp\left(- kC - (1/m)l \gamma^{\tau+1} \right )\\
= \;& \exp\left(- C - (1/m)l \gamma^{\tau+1} \right )
\end{align*}
Thus, Eq.\eqref{eq:letter_prob_decay} holds at level $\tau+1$. Induction completes the proof.
\end{proof}

\begin{theorem}
For $k \geq 2$ and $m \geq 3$, there exists $\delta_0 \equiv \delta_0(m,k)>0$, and a node-oblivious decision rule vector, such that
the following is true: For any $\delta < \delta_0$, we have
\begin{align}
 \prob\big [\sigma_\root\neq s \big ] &\leq  \exp \left \{ - \frac{m-1}{2m} \big \{k  \left (1-1/m\right ) \big \}^t \right \} \nonumber\\
& =  \exp \left \{- \frac{m-1}{2m}\, n^\rho \right \}
\end{align}
with $\rho \equiv 1 + \log(1-1/m)/\log k$.
\label{thm:good_decision_rule_exists}
\end{theorem}

\begin{proof}
The theorem follows from Lemma \ref{lemma:good_decision_decay} and the root decision rule Eq.~\eqref{eq:root_good_dec_rule}.

Assume $H_0$. For every $\sigma_{\partial \root}= (-(m-1)/2 + l_1, -(m-1)/2 + l_2, \ldots, -(m-1)/2 + l_k )$ such that $S_\root \geq 0$, we have
$\sum_{j=1}^k l_j \geq k (1-1/m)(m-1)/(2m) $. From Lemma \ref{lemma:good_decision_decay}(i),
\begin{align}
\prob_0 (\sigma_{\partial \root}|H_0) &\leq \exp\left(- kC - (1/m)\gamma^{t-1}\sum_{j=1}^k l_j \right ) \nonumber\\
&\leq \exp\left(- kC - (m-1) \gamma^{t}/(2m) \right )\, ,
\end{align}
where $\gamma \equiv k  \left (1-1/m\right )$ and $C\equiv k \log m/ (k-1) $.
Obviously, there are at most $m^k$ such $\sigma_{\partial \root}$. It follows that
\begin{align*}
\prob_0 (\sigma_\root = 1|H_0) &\leq m^k \exp\left(- kC - (m-1) \gamma^{t}/(2m) \right )\nonumber\\
&= \exp\left(- C - (m-1) \gamma^{t}/(2m) \right )\, .
\end{align*}

Similarly, we can show
\begin{align*}
\prob_1 (\sigma_\root = 0|H_1) &\leq \exp\left(- C - (m-1) \gamma^{t}/(2m) \right )
\end{align*}

Combining, we arrive at
\begin{align*}
\prob (\sigma_\root \neq s) & \leq \exp\left(- C - (m-1) \gamma^{t}/(2m) \right )
\end{align*}
Recall that $C>0$. Thus, we have proved the result.
\end{proof}

\subsection{Subexponential decay of error probability}
\label{subsec:subexp_decay_general}

Define $n \equiv k^t$, i.e., $n$ is the number of private signals received, one at each leaf. The scheme presented in the previous section allows us to achieve error probability that decays like
$\exp(-\Omega(\{k  \left (1-1/m\right ) \}^t)) = \exp(-\Omega(n^\rho))$, where $\rho = 1 + \log(1-1/m)/\log k \approx 1- 1/(m \log k)$ for $m \gg 1$. In this section we show that under appropriate assumptions, error probability that decays exponentially in $n$, i.e., $\exp(- \Theta(n))$, is not achievable with node-oblivious rules.

In this section we call the letters of the message alphabet $\cM= \{1,2, \ldots, m\}$. 
%We allow arbitrary private signal distributions $\prob_0(x_i)$ and $\prob_1(x_i)$ at the leaf nodes that are absolutely continuous with respect to each other. 
For simplicity, we consider only deterministic node-oblivious rules, though our results and proofs extend easily to randomized rules.

We define here a directed graph $\cG$ with vertex set $\cM$ and edge set $\cE$ that we define below. We emphasize that $\cG$ is distinct from the tree on which information aggregation is occurring. There is a directed edge from node $\mu_i \in \cM$ to node $\mu_j \in \cM$ in $\cG$ if there exists $\balf \in \cM^k$ such that $\mu_j$ appears at least once in $\balf$ and $f(\balf) = \mu_i$.
Informally, $(\mu_i, \mu_j) \in \cE$ if $\mu_i$ can be `caused' by a message vector received from children that includes $\mu_j$. We call $\cG$ the \emph{dependence graph}.

We make the following irreducibility assumptions on the node-oblivious decision rule vectors $(f,g,h)$ under consideration (along with leaf and root decision rules).

\begin{assumption}
\label{ass:strongly_connected}
The dependence graph $\cG$ is strongly connected. In other words, for any $\mu_i \in \cM$ and $\mu_j \in \cM$ such that $\mu_j \neq \mu_i$, there is a directed path from $\mu_i$ to $\mu_j$ in $\cG$.
\end{assumption}

\begin{assumption}
There exists a level $\taup>0$ such that for node $i$ at level $\taup$, we have $\prob_0(\sigma_i=\mu) > 0$ for all $\mu \in \cM$.
\label{ass:all_letters_+ve_prob}
\end{assumption}

Note that $\prob_0(\sigma_i=\mu) > 0$ implies $\prob_1(\sigma_i=\mu) > 0$ by absolute continuity of $\prob_0(x_i)$ w.r.t. $\prob_1(x_i)$.

\begin{assumption}
There exists $\mu_- \in \cM$, $\mu_+ \in \cM$, $\eta>0$ and $\tau_{*}$ such that, for all $ \tau > \taud$ the following holds: For node $i$ at level $\tau$, we have $\prob_0(\sigma_i = \mu_-)>\eta$ and $\prob_1(\sigma_i = \mu_+)>\eta$.
\label{ass:one_letter_dominant}
\end{assumption}
In other words, we assume there is one `dominant' message under each of the two possible hypothesis.

It is not hard to verify that for $k \geq 2$, $m \geq 3$ and $\delta < \delta_0(m,k)$ (where $\delta_0$ is same as in Lemma \ref{lemma:good_decision_decay} and Theorem \ref{thm:good_decision_rule_exists}), the scheme presented in the previous section satisfies all four of our assumptions. In other words, the assumptions are all satisfied in the regime where our scheme has provably good performance.

\begin{definition}
Consider a directed graph $\cG = (\cV, \cE)$ that is strongly connected. For $u,v \in \cV$, let $d_{uv}$ be the length of the shortest path from $u$ to $v$. Then the \emph{diameter} of
$\cG$ is defined as 
\vskip-15pt
\begin{align*}
\textup{diameter}(\cG) \equiv \max_{u\in \cV} \max_{v\in \cV, v \neq u} d_{uv} \; .
\end{align*}
\end{definition}
\vskip5pt

\begin{theorem}
\label{thm:subexp_decay}
Fix $m$ and $k$. Consider any node-oblivious decision rule vector $(f,g,h)$ such that Assumptions \ref{ass:strongly_connected}, \ref{ass:all_letters_+ve_prob} and \ref{ass:one_letter_dominant} are satisfied. Let $d$ be the diameter of the dependence graph $\cG$. Then, there exists $C \equiv C(f,m,k)< \infty$ such that we have
\begin{align}
\prob\big [\sigma_\root\neq s \big ] \geq \exp \left \{- C n^{\brho} \right \}\, ,
\end{align}
where $\brho \equiv 1 + \frac{\log(1- k^{-d})}{d\log k} < 1$.
\end{theorem}
\vskip4pt

Now $\cG$ has $m$ vertices, so clearly $d \leq m-1$. The following corollary is immediate.
\begin{corollary}
\label{coro:subexp_decay}
Fix $m$ and $k$. Consider any node-oblivious decision rule vector $(f,g,h)$ such that Assumptions \ref{ass:strongly_connected}, \ref{ass:all_letters_+ve_prob} and \ref{ass:one_letter_dominant} are satisfied. Then, there exists $C \equiv C(f,m,k)< \infty$ such that we have
\begin{align}
\prob\big [\sigma_\root\neq s \big ] \geq \exp \left \{- C n^{\rho} \right \}\, ,
\end{align}
where $\rho \equiv 1 + \frac{\log(1- k^{-(m-1)})}{(m-1)\log k} < 1$.
\end{corollary}
\vskip4pt

Thus, we prove that under the above irreducibility assumptions, the error must decay subexponentially in the number of private signals available at the leaves.

\begin{remark}
\label{rem:mus_vec_maps_to_itself}
We have $\prob_0(\sigma_{\partial \root} = (\mu_-, \mu_-, \ldots,\mu_-))> \eta^{k}$. It follows that we must have $f_\root(\mu_-, \mu_-, \ldots,\mu_-)=0$ (else the probability of error is bounded below by $\eta^k/2$ for any $t$). Similarly, we must have $f_\root(\mu_+, \mu_+, \ldots,\mu_+)=1$. In particular, $\mu_- \neq \mu_+$.
\end{remark}

\begin{lemma}
If Assumption \ref{ass:all_letters_+ve_prob} holds, then for a node $i$ at any level $\tau > \taup$, we have $\prob_0(\sigma_i=\mu) > 0$ for all $\mu \in \cM$.
\label{lemma:all_letters_+ve_prob_after_tau0}
\end{lemma}
\begin{proof}
It follows from Assumption \ref{ass:all_letters_+ve_prob} that for any $\mu \in \cM$, there is some $\balf \in \cM^k$ such that $f(\balf_\mu) = \mu$. We prove the lemma by induction on the level $\tau$. Let
\begin{align*}
\cS_\tau \equiv \textup{For node $i$ at level $\tau$, $\prob_0(\sigma_i=\mu) > 0$ for all $\mu \in \cM$.}
\end{align*}
By assumption, $\cS_{\taup}$ holds. Suppose $\cS_{\tau}$ holds. Consider node $i$ at level $\tau+1$. Consider any $\mu \in \cM$. By inductive hypothesis, we have $\prob_0(\sigma_{\partial i}= \balf_\mu)>0$. It follows that $\prob_0(\sigma_i=\mu) > 0$. Thus, $\cS_{\tau+1}$ holds. 
\end{proof}

Lemma \ref{lemma:subexp_decay} can be thought of as a quantitative version of Lemma \ref{lemma:all_letters_+ve_prob_after_tau0}, showing that the probability of the least frequent message decays subexponentially.

\begin{lemma}
Suppose Assumptions \ref{ass:strongly_connected}, \ref{ass:all_letters_+ve_prob} and \ref{ass:one_letter_dominant} are satisfied. Fix $s \in \{0,1\}$. Consider a node $i$ at level $\tau$. Define $\zeta_\tau \equiv \min_{\mu \in \cM} \prob(\sigma_i = \mu|H_s)$. Let $\tau_* = \max(\taup, \taud)$ (cf. Assumptions \ref{ass:all_letters_+ve_prob}, \ref{ass:one_letter_dominant}). Let $d= \textup{diameter}(\cG)$. There exists $C' \equiv C'(f,m,k)< \infty$ such that for any $a \in \naturals \cup \{0\}$ and $b \in \{0,1, \ldots, d-1\}$, we have,
\begin{align}
\zeta_{\tau_*+ad+b} \geq  \exp \left \{ -C' (k^{d}-1)^a \right \}
\end{align}
\label{lemma:subexp_decay}
\end{lemma}
\begin{proof}
Assume $H_0$ holds, i.e. $s=0$. The proof for $s=1$ is analogous.

We prove that, in fact, the following stronger bound holds:
\begin{align}
-\log(\zeta_{\tau_*+ad+b}) \leq C' (k^{d}-1)^a - \log (1/\eta)/(k^{d}-2) \, .
\label{eq:log_ptau_strong_bound}
\end{align}

We proceed via induction on $a$. First consider $a=0$. Consider a node $i$ at level
$\tau_*+b$ for $b \in \{0,1, \ldots, d-1\}$. Consider the descendants of node $i$ at level $\tau_*$. For any $\mu \in \cM$, we know from Lemma \ref{lemma:all_letters_+ve_prob_after_tau0} that there must be \textit{some} assignment of messages to the descendants, such that $\sigma_i = \mu$. It follows that
\begin{align}
\zeta_{\tau_*+b} \geq \zeta_{\tau_*}^{k^b}
\end{align}
Thus, choosing $C'= k^{d-1}(- \log \zeta_{\tau_*}) +  \log (1/\eta)/(k^{d}-2)$, we can ensure
that Eq.~\eqref{eq:log_ptau_strong_bound} holds for $a=0$ and all $b \in \{0,1, \ldots, d-1\}$.

Now suppose Eq.~\eqref{eq:log_ptau_strong_bound} holds for some $a \in \naturals \cup \{0 \}$. Consider a node $i$ at level $\tau_*+(a+1)d+b$. Let $\cD$ be the set of descendants of node $i$ at level $\tau_*+ad+b$. Note that $|\cD|=k^{d}$. Consider any $\mu \in \cM$. By Assumption \ref{ass:strongly_connected}, there is a directed path in $\cG$ of length at most $d$ going from $\mu$ to $\mu_-$. By Remark \ref{rem:mus_vec_maps_to_itself}, we know that $(\mu_-, \mu_-) \in \cE$. It follows that there is a directed path in $\cG$ of length \textit{exactly} $d$ going from $\mu$ to $\mu_-$. Thus, there must be an assignment of messages $\sigma_{\cD}$  to  nodes in $\cD$, including at least one occurrence of $\mu_-$, such that $\sigma_i = \mu$. Using Assumption \ref{ass:one_letter_dominant}, we deduce that
\begin{align*}
\zeta_{\tau_*+(a+1)d+b}  \geq\eta \zeta_{\tau_*+ad+b}^{k^{d}-1}
\end{align*}
Rewriting as
\begin{align*}
-\log & \,\zeta_{\tau_*+(a+1)d+b}  \leq \\
&({k^{d}-1})(-\log \zeta_{\tau_*+ad+b}) + \log (1/\eta) \, ,
\end{align*}
and combining with Eq.~\eqref{eq:log_ptau_strong_bound}, we obtain
\begin{align*}
-\log&(\zeta_{\tau_*+(a+1)d+b}) \leq \\
&C' (k^{d}-1)^{a+1} - \log (1/\eta)/(k^{d}-2) \, .
\end{align*}
Induction completes the proof.
%HERE
\end{proof}

Theorem \ref{thm:subexp_decay} follows.

\begin{proof}[Proof of Theorem \ref{thm:subexp_decay}]
%In Assumption \ref{ass:one_letter_dominant}, if $\mu_- = \mu_+$ then, for $t > \taud$, we will have $\prob(\sigma_{\partial \root} = (\mu_-, \mu_-, \ldots,\mu_-)|H_s)> 2^{-k}$ for both $s=0$ and $s=1$. Without loss of generality, suppose $f_\root(\mu_-, \mu_-, \ldots,\mu_-) = 0$. Then, $\prob_1(\sigma_\root =0) > 2^{-k}$, implying $-\log \prob\big [\sigma_\root\neq s \big ]\leq (k+1) \log(2)$. So Theorem \ref{thm:subexp_decay} holds trivially. So we assume $\mu_- \neq \mu_+$ henceforth.
Assume $H_0$.
From Lemma \ref{lemma:subexp_decay},
\begin{align*}
 \prob_0(\sigma_{\partial \root} = (\mu_+, \mu_+, \ldots,\mu_+)) &\geq \exp \left \{- C' k^{\brho ad} \right \} \\
 &\geq  \exp \left \{-C n^{\brho}\right \}
\end{align*}
for $C\equiv C' k^{\brho(\tau_*+d-1)}$.
It follows that
\begin{align}
\prob_0(\sigma_\root = 1) \geq  \exp \left \{-C n^\brho \right \}\, .
\end{align}

Similarly,
\begin{align}
\prob_1(\sigma_\root = 0) \geq \exp \left \{- C n^\brho \right \}\, .
\end{align}
The result follows.
\end{proof}

\begin{remark}
For the scheme presented in Section \ref{subsec:node-oblivious_efficient_scheme}, we have $d \approx \log_k m$, where $d = \textup{diameter}(\cG)$. For any $\eps>0$,  Theorem \ref{thm:subexp_decay} provides a lower bound on error probability with $\brho \leq 1 - C_1/m^{1+\eps}$ for some
$C_1\equiv C_1(k, \eps) >0$. This closely matches the $m$ dependence of the upper bound on error probability we proved in Theorem \ref{thm:good_decision_rule_exists}.
\end{remark}

\subsection{Discussion of the irreducibility assumptions}

We already mentioned that the efficient node-oblivious rule presented in Section \ref{subsec:node-oblivious_efficient_scheme} satisfies all of Assumptions \ref{ass:strongly_connected}, \ref{ass:all_letters_+ve_prob} and \ref{ass:one_letter_dominant}.
Moreover, it is natural to expect that similar schemes based on propagation of quantized likelihood ratio estimates should also satisfy our assumptions. In this section, we further discuss our assumptions taking the cases of binary and ternary messages as examples. 

\subsubsection{Binary messages}

Binary messages are not the focus of Section \ref{subsec:subexp_decay_general}. However, we present here a short discussion of  Assumptions 1, 2 and 3 in the context of binary messages for illustrative purposes. 

\noindent{\bf Claim:} If $m=2$, each of the irreducibility assumptions \emph{must} be satisfied by any node-oblivious rule for which error probability decays to $0$ with $t$.

\noindent Proof of Claim: Call the messages $\cM=\{0,1\}$. Consider a node-oblivious decision rule vector $(f,g,h)$ such that error probability decays to $0$ with $t$. Then $g$ cannot be a constant function (e.g., identically $0$), since this leads to $\pe \geq 1/2$.
%We denote the probability of error in a tree of depth $t$ by $\pe^{(t)}$.

Suppose Assumption \ref{ass:strongly_connected} is violated. Without loss of generality, suppose $(0,1) \notin \cE$. Then $f(\balf) = 1$ for all $\balf \neq (0, 0, \ldots, 0)$. It follows that for node $i$ at level $\tau$, we have
\begin{align}
\prob_s(\sigma_i=0) \leq  \exp(-\Theta(k^\tau)) \stackrel{t \rightarrow \infty}{\longrightarrow} 0 \, ,
\end{align}
for both $s=0$ and $s=1$. In particular, $\pe$ is bounded away from $0$. This is a contradiction.

Suppose Assumption \ref{ass:all_letters_+ve_prob} is violated. Then, wlog, all nodes at level $1$ transmit the message $1$ almost surely, under either hypothesis. Thus, all useful information is lost and $\pe \geq 1/2$. This is a contradiction.

Finally, we show that Assumption \ref{ass:one_letter_dominant} must hold as well. Define $\xi_\tau \equiv \prob_0(\sigma_i = 0)$ for node $i$ at level $t$.
Wlog, suppose $\xi_\tau \geq 1/2$ occurs infinitely often. Then we have
$h(0, 0, \ldots, 0)=0$, else $\pe \geq 2^{-k-1}$ for infinitely many $t$.
Define $\hxi_\tau \equiv \prob_1(\sigma_i = 0)$ for node $i$ at level $t$. If $\hxi_\tau \geq 1/2$ occurs infinitely often, then it follows that $\prob_1(\sigma_{\partial \root} = (0,0, \ldots, 0) )\geq 2^{-k}$ and hence $\prob_1(\sigma_{\root} = 0 )\geq 2^{-k}$ occur for infinitely many $t$. So we can have $\hxi_\tau \geq 1/2$ only finitely many times. Also, $h(1, 1, \ldots, 1)=1$ must hold. It follows that $\xi_\tau < 1/2$ occurs only finitely many times. Thus, Assumption \ref{ass:one_letter_dominant} holds with $\eta = 1/2$.
%, else $\lim \inf_{t \rightarrow \infty} \pe^{(t)} \geq 2^{-k-1}$.

\subsubsection{Ternary messages}
By Theorem \ref{thm:good_decision_rule_exists}, the scheme presented in Section \ref{subsec:node-oblivious_efficient_scheme} achieves $\pe = \exp \left \{-\Omega(\{2k/3\}^t \right \}$ in the case of ternary messages.

We first show that if Assumption \ref{ass:all_letters_+ve_prob} is violated, then $\pe = \exp \left \{-O(\{(k+1)/2\}^t)\right \}$. If Assumption \ref{ass:all_letters_+ve_prob} does not hold, then only at most two letters are used at each level. It follows that we can have a (possibly node-dependent) scheme with binary messages that is equivalent to the original scheme at levels $1$ and higher. Our lower bound on $\pe$ then follows from Theorem \ref{thm:binary_lower_bound}. Thus, even in the best case, performance is significantly worse than the scheme presented in Section \ref{subsec:node-oblivious_efficient_scheme}. So a good scheme for ternary messages must satisfy Assumption \ref{ass:all_letters_+ve_prob}.

Now consider Assumption \ref{ass:strongly_connected}. Let $\cM=\{-1, 0, 1\}$. Suppose Assumption \ref{ass:strongly_connected} is violated. Then wlog, there is no path from letter $0$ to one of the other letters. It follows that under either hypothesis, we have $\prob_s(\sigma_i = 0) = \exp \left \{- \Omega(k^\tau)\right \}$ for node $i$ at level $\tau$. Thus, the letter $0$ occurs with exponentially small probability, irrespective of $s$. This should essentially reduce, then, to the case of binary messages, and we expect performance to be constrained as above.

Finally, consider Assumption \ref{ass:one_letter_dominant}. We cannot have $h(\mu, \mu, \mu)=0$ for all $\mu \in \cM$, since that will lead to $\prob_1(\sigma_\root \neq s) \geq 1/9$ for all $t$. Similarly, we can also exclude the possibility $h(\mu, \mu, \mu)=1$ for all $\mu \in \cM$. Wlog, suppose $h(-1,-1,-1)=0$ and $h(1,1,1)=1$. Now consider the problem of designing a good aggregation protocol. By the above, we must have $\prob_1(\sigma_i =-1 )$ and $\prob_0(\sigma_i =1 )$, for node $i$ at level $\tau$, to each converge to 0 with increasing $\tau$. Further, it appears natural to use the message $\mu=0$ with an interpretation of `not sure' in such a situation. We would then like the probability of this intermediate symbol to decay with $\tau$, or at least be bounded in the limit, i.e., $\lim \sup_{\tau \rightarrow \infty} \prob_s(\sigma_i =0 ) <1$ for each possible $s$. If this holds, we immediately have Assumption \ref{ass:one_letter_dominant} (with $\mu_-=-1$ and $\mu_+=1$).

\subsubsection{Need for assumptions}
We argued above that our irreducibility assumptions are quite reasonable in various circumstances.
In fact, we expect the assumptions to be a proof artifact, and conjecture that a subexponential convergence bound holds for general node-oblivious rules. A possible approach to eliminate our assumptions would be to prune the message alphabet $\cM$, discarding letters that never appear, or appear with probability bounded by $\exp(-\Omega(k^t))$ (because they require descendants from a strict subset of $\cM$).

\end{document}